\documentclass[12pt]{article}
\usepackage{epsf,epsfig,amsfonts,amsgen,amsmath,amstext,amsbsy,amsopn,amsthm,lineno}
\usepackage{color}
\usepackage{tikz}
\usepackage{appendix}
\usepackage{xspace}
\newtheorem{theorem}{Theorem}
\newtheorem{lemma}{Lemma}
\theoremstyle{definition}
\newtheorem{definition}{Definition}
\newtheorem{corollary}{Corollary}

\usepackage{xcolor}
\usepackage{tcolorbox}

\newcommand{\ff}{\textsc{FirstFit}\xspace}
\newcommand{\ffp}{\textsc{FirstFitPredictions}\xspace}
\newcommand{\ffpshort}{\textsc{FFP}\xspace}
\newcommand{\rffp}{\textsc{RobustFirstFitPredictions}\xspace}
\newcommand{\rffpshort}{\textsc{RFFP}\xspace}
\newcommand{\bcm}{\textsc{BiColorMax}\xspace}

\providecommand{\keywords}[1]
{
  \small
  \textbf{\textit{keywords---}}#1
}

\usepackage{enumitem}

\title{Online  Graph Coloring with Predictions}

\author{Antonios Antoniadis , Hajo Broersma, Yang Meng
\\[2mm]
\small Faculty of Electrical Engineering, Mathematics and Computer Science \\
\small  University of Twente
}

\begin{document}

\maketitle

\begin{abstract}
We introduce learning augmented algorithms to the online graph coloring problem.
Although the simple greedy algorithm \ff is known to perform poorly in the worst case, we are able to establish a relationship between the structure of any input graph $G$ that is revealed online and the number of colors that \ff uses for $G$. Based on this relationship, we propose an online coloring algorithm \ffp that extends \ff while making use of machine learned predictions. We show that \ffp is both \emph{consistent} and \emph{smooth}. Moreover, we develop a novel framework for combining online algorithms at runtime specifically for the online graph coloring problem. Finally, we show how this framework can be used to robustify \ffp by combining it with any classical online coloring algorithm (that disregards the predictions).
\end{abstract}

\thispagestyle{empty}

\pagebreak
\setcounter{page}{1}

\section{Introduction}
Before we will properly define the concepts we use throughout this paper in Section~\ref{pre}, let us start with a less formal introduction to the subject, together with some background and motivation.

Graph coloring is a central topic within graph theory that finds its origin in the notorious Four Color Problem, dating back to 1852. Since then graph coloring has developed into a mature research field with numerous application areas, ranging from scheduling~\cite{Marx_2004, Manjula} and memory allocation~\cite{Barenboim} to robotics~\cite{DEMANGE}. In the online version of the problem, the vertices of a (usually unknown) graph arrive online one by one, together with the adjacencies to the already present vertices. Upon the arrival of each vertex $v$, a color has to be irrevocably assigned to $v$. The goal is to obtain a proper coloring, that is a coloring in which no two adjacent vertices have the same color. The challenge is to design a coloring strategy which keeps the total number of assigned colors relatively low.

In this paper, we introduce the online graph coloring problem to learning augmented algorithms. We assume that alongside the arrival of each vertex, the algorithm also obtains a prediction $P(v)$ (of unknown quality) on the color that should be assigned to $v$. These predictions may be used by an algorithm to obtain a coloring which uses fewer colors, if the predictions turn out to be relatively accurate. At the same time, the algorithm should maintain a worst-case guarantee to safeguard against the case in which the predictions are inaccurate.

Graph coloring is notoriously hard, already in the \emph{offline setting}, where the whole input graph is known in advance. Let $\chi(G)$ denote the \emph{chromatic number} of a graph $G$, that is the minimum number of distinct colors needed to obtain a proper coloring of $G$. A straightforward observation is that $\chi(G)$ is greater than or equal to $\omega(G)$, which is defined as the cardinality of a maximum clique in $G$. However, there even exist triangle-free graphs (so with no cliques of cardinality 3) with an arbitrarily large chromatic number. In case $\chi(H)=\omega(H)$ for every induced subgraph $H$ of a graph $G$, then $G$ is called a \emph{perfect graph}. Perfect graphs are known to be $\chi(G)$-colorable in polynomial-time via semidefinite programming~\cite{CHUDNOVSKY}. An interesting special case of perfect graphs is the class of bipartite graphs. These graphs admit a proper coloring using only $2$ colors, and in the offline setting such a $2$-coloring can be computed in linear time, for example via breadth first search. However, in general it is an NP-complete problem to decide whether a given graph admits a proper coloring using $k$ colors, for any fixed $k\ge 3$~\cite{Karp1972}. With respect to approximation algorithms, there is a polynomial-time algorithm using at most $O(n(\log \log n)^2/(\log n)^3)\cdot \chi(G)$~\cite{HALLDORSSON199319} colors for a graph on $n$ vertices, and it is NP-hard to approximate the chromatic number within a factor $n^{1-\epsilon}$ for all $\epsilon >0$~\cite{Zuckerman07}.

The problem becomes even more challenging in the \emph{online setting} where, as mentioned, vertices arrive online one by one and an algorithm has to irrevocably assign a color to each vertex upon its arrival -- while only having knowledge of the subgraph revealed so far. Any online algorithm may require at least $(2n/(\log n)^2)\chi(G)$ colors in the worst case~\cite{HALLDORSSON1994163}, where $n$ is the number of vertices of the input graph -- and this is true even for bipartite graphs (so with chromatic number $2$).
Restricting the input graphs even further, for instance to $P_6$-free bipartite graphs (containing no path on six vertices as an induced subgraph), no online algorithm can guarantee a coloring with constantly many colors~\cite{[1]}. In such cases, bounding the number of colors used from above by a function of the chromatic number is not feasible. To this end, a line of research has focused on developing so-called \emph{online-competitive algorithms}. Applied to any graph $G$, such online algorithms are guaranteed to produce a proper coloring, the number of colors of which is bounded from above by a function of the number of colors used by the \emph{best possible online algorithm} for $G$. As an example, in the previously mentioned setting of $P_6$-free bipartite graphs, there exists an online algorithm that uses at most
twice as many colors as the best possible online algorithm~\cite{[7],[10]}. A good source for more information regarding online graph coloring is the following book chapter due to
Kierstead~\cite{Kierstead96}. We will come back to this in Section~\ref{sec:graph-classes}, where we review and utilize several known results, including more recent work.

Probably the conceptually simplest online algorithm for graph coloring is the greedy algorithm, which is known as \ff. Suppose we have a total order over the set of available colors. Then upon arrival of each vertex, \ff assigns to it the smallest color according to that order, among the ones that maintain a proper coloring. This algorithm has been extensively analyzed in the literature, also for particular graph classes~\cite{[1],[16]}. Although \ff performs well for many practically relevant inputs, it can be very sensitive to the order in which the vertices of $G$ are revealed. For instance, in a popular example where $G$ is a complete balanced bipartite graph $K_{n,n}$ minus a perfect matching, there is a specific permutation on the arrival of the vertices of $G$ for which \ff requires $n$ colors (whereas $\chi(G)=2$)~\cite{Johnson}.

That \ff performs well in some practical scenarios, can be attributed to the fact that real-world graph coloring instances rarely resemble worst-case
inputs. It is often the case that either the structure of the input graph or the permutation in which the vertices are revealed can be exploited by a heuristic or a machine-learning approach in order to yield reasonably good colorings, despite the inherent worst-case difficulty of the problem. However, and not too surprisingly, such approaches tend to come without a worst-case guarantee. In the (hopefully rare) cases where the input diverges substantially from the expected structure, the resulting coloring could be arbitrarily poor.

In this paper, we design an algorithm that incorporates predictions of unknown quality obtained by such a machine-learned approach. It produces a relatively good coloring in case the predictions turn out to be accurate, while at the same time providing a worst-case guarantee comparable to the best classical online algorithm that does not make use of the predictions. Our work falls within the context of \emph{learning augmented algorithms}.

Learning augmented algorithms is a relatively new and very active field. The main goal is to develop algorithms combining the respective advantages of machine-learning approaches and classical worst case
algorithm analysis. A plethora of online problems have already been investigated through the learning augmented algorithm lens, including for example, caching~\cite{DBLP:conf/icml/LykourisV18,antoniadis2022paging}, facility location~\cite{cohen2023general}, ski-rental~\cite{DBLP:conf/nips/PurohitSK18,shin2023improved}, or
various scheduling problems~\cite{DBLP:conf/nips/PurohitSK18,DBLP:conf/soda/LattanziLMV20}, to name just a few. To the best of our knowledge, learning augmented algorithms have not been studied for graph coloring problems to date. For a more extensive discussion on learning augmented algorithms, we
refer the interested reader to a recent
survey~\cite{MitzenmacherV22}.

A common approach to developing learning augmented algorithms is to design an algorithm that attempts to follow the predictions, in some sense. At the same time, this algorithm should be robustified by appropriately
combining it with a classical algorithm that disregards the predictions. At a high level, this is our approach for graph coloring as well. However, the nature of the problem poses several novel challenges. First of all, already assigned colors may significantly restrict the choice of colors for the next and future assignments of
the algorithm. This is in contrast to settings where an algorithm can, at some cost, move to any arbitrary configuration, for instance, in problems with an underlying metric. The fact that it is not possible for an algorithm to move to any possible configuration also rules out a robustification approach by combining algorithms in an experts-like setting. See~\cite{AntoniadisCE0S20} for more information. Secondly, existing online algorithms for online graph coloring do not possess a particular monotonicity property that tends to be a crucial ingredient in robustifying algorithms for other problems. In particular, it is possible that running an algorithm only on the graph induced by a suffix of the input permutation requires significantly more colors than running the same algorithm on the graphs of the complete input permutation. This further complicates the robustification, since one can not directly use a classical algorithm as a fall back option upon recognizing that the predictions are of insufficient quality.

\subsection{Preliminaries}\label{pre}
In this section, we formally define the problem setting and its associated prediction model. We start with the concepts related to (offline) graph coloring.

\begin{definition}[Graph coloring~\cite{[4]}]
A \emph{$k$-vertex coloring}, or simply a \emph{$k$-coloring}, of a graph $G$ is a mapping $\phi:V(G) \rightarrow S$, where $S$ is a set of $k$ colors. A $k$-coloring
is proper if no two adjacent vertices are assigned the same color. A graph is \emph{$k$-colorable} if it admits a proper $k$-coloring. The minimum $k$ for which a graph $G$ is $k$-colorable is called its \emph{chromatic number}, denoted by $\chi(G)$. An {\em optimal coloring} of $G$ is a proper $\chi(G)$-coloring.
\end{definition}

Online graph coloring describes the setting in which the vertices of $G$ arrive one by one in an online fashion and have to be irrevocably and properly colored upon arrival.

\begin{definition}[Online graph coloring~\cite{[3]}]
An \emph{online graph} $(G,\pi)$ is a graph $G$
together with a permutation $\pi = v_1,v_2,\ldots,v_n$ of $V(G)$. An \emph{online coloring algorithm} takes an online graph $(G,\pi)$ as input and produces a proper coloring of $V(G)$, where the color of a vertex $v_{i}$ is chosen from a \emph{universe} $\mathcal{U}$ of available colors and the choice depends only on the subgraph of $G$ induced by $\{v_{1},\ldots,v_{i}\}$ and the colors assigned to $v_{1},\ldots,v_{i-1}$, for $1 \le i \le n$.
\end{definition}

Throughout the paper, and unless otherwise specified, algorithm refers to a \emph{deterministic} algorithm.

We next define the notions of \emph{competitiveness, online competitiveness} and \emph{competitive ratio}, which are used to evaluate the performance of algorithms for online graph coloring.

\begin{definition}[Competitiveness~\cite{[1]}, online competitiveness~\cite{[8]} and competitive ratio~\cite{[6]}]
Let $AOL(G)$ be the set of all online coloring algorithms for a graph $G$ and let $\Pi(G)$ be the set of all permutations of $V(G)$. For an algorithm $A\in{AOL(G)}$ and a permutation $\pi\in{\Pi(G)}$, the number of colors used by $A$ when $V(G)$ gets revealed according to $\pi$ is denoted by $\chi_A(G,\pi)$. The $A$-chromatic number of $G$ is the largest number of colors used by the online algorithm $A$ for the graph $G$, denoted by $\chi_A(G)$. That is,
\begin{center}
$\chi_A(G)=\mathop{\max}\limits_{\pi\in{\Pi(G)}}\chi_{A}(G,\pi)$.
\end{center}
For a graph $G$, the online chromatic number $\chi_{OL}(G)$ is the minimum number of colors used for $G$, over all algorithms of $AOL(G)$. That is,
\begin{center}
$\chi_{OL}(G)=\mathop{\min}\limits_{A\in{AOL(G)}}\chi_{A}(G)$.
\end{center}
Let $\mathcal{G}$ be a family of graphs and $AOL(\mathcal{G})$ be the set of online algorithms for $\mathcal{G}$. For some $A\in{AOL(\mathcal{G})}$, if there exists a function such that $\chi_{A}(G)\leq{f(\chi(G))}$, (resp. $\chi_{A}(G)\leq{f(\chi_{OL}(G))}$) holds for every $G\in{\mathcal{G}}$, then $A$ is \emph{competitive} (resp. \emph{online competitive}) on $\mathcal{G}$.

Furthermore, the \emph{competitive ratio} of an algorithm $A\in{AOL(\mathcal{G})}$ over a class of graphs $\mathcal{G}$ is the maximum of $\frac{\chi_{A}(G)}{\chi(G)}$ for all $G\in{\mathcal{G}}$.
\end{definition}

We note that the notion of competitiveness used here follows the literature on online graph coloring and contrasts the definition commonly used for other online problems, where an algorithm is said to be competitive if it attains a constant competitive ratio. We complete this section by presenting the basic definitions associated with the setting, in which we involve predictions on the colors.

\paragraph{Predictions and prediction error.}

Here we assume that alongside the disclosure of each vertex $v$, the algorithm also obtains a prediction $P: V(G)\rightarrow \mathcal{U}$ on the color of $v$, where $\mathcal{U}$ is the set of available colors. These predictions are aimed at obtaining a reasonable coloring. They may stem from a machine-learning approach based on past inputs or training data, or from a simple heuristic known to perform well in practice. The quality of the obtained predictions is measured by means of a \emph{prediction error}. This \emph{prediction error} is defined naturally to be the (smallest) number of vertices that obtained wrong predictions.

\begin{definition}[Prediction error]
 Given an online graph $(G,\pi)$, let $\mathcal{O}(G)$ be the set of all optimal colorings of $G$, where $O\in\mathcal{O}$ assigns color $O(v)\in \mathcal{U}$ to vertex $v$. Then the \emph{prediction error} for online graph $(G,\pi)$ is given by
\begin{align*}
    \eta(G) = \min_{O\in\mathcal{O}(G)}\Sigma_{v\in V(G)} |\{P(v)\}\setminus \{O(v)\}|.
\end{align*}

In the following, we drop the dependence on $G$ when the underlying online graph is clear from the context. We use the notation $(G,\pi,P)$ to refer to an \emph{online graph with predictions}, where $G$ is the underlying graph, $\pi$ the permutation in which
$V$ is revealed, and $P$ the set of associated predictions.
\end{definition}

Following the literature, we say that an algorithm is \emph{$\alpha$-consistent} if it attains a competitive ratio of $\alpha$ in the case that the predictions are perfect ($\eta=0$), and \emph{robust} if it independently of the prediction error always obtains a competitive ratio within a constant factor of that of the best known classical online algorithm. Furthermore, we say that an algorithm is \emph{smooth} if its competitive ratio degrades at a rate that is at most linear in the prediction error. Note that the notion of robustness also extends to the case where the best known classical online algorithm $A$ is ``only" online competitive. In this scenario, any algorithm that is guaranteed to use at most a constant number of colors more than what $A$ uses is \emph{robust}.

We note that one can trivially obtain an optimal coloring when the predictions are perfect (in other words when $\eta = 0$) by just coloring each vertex $v$ with color $P(v)$ upon arrival. This already is a $1$-consistent algorithm. However, when the obtained predictions are only slightly off, this algorithm may not even be a valid algorithm for online graph coloring. Indeed, consider the case where only one vertex $v$ receives a wrong prediction $P(v)$ but is adjacent to a vertex $u$ with $P(u)=P(v)$.

\subsection{Our Contribution}
Our first contribution lies in establishing a relationship between the structure of an online graph $(G,\pi)$ and the amount of colors used by \ff for $G$. We emphasize that this result is independent of predictions and might be of broader interest. More specifically, in Section~\ref{sec:structural} we show that if \ff uses $x$ colors for $G$, then there exists a set $V'\subseteq V$ of vertices of size $|V'|=x+q$, for some $0\le q\le x-2$, such that $V'$ can be partitioned into $q+1$ non-trivial subsets, each of which is a \emph{clique} in $G$ (that is,
each subset consists of at least two vertices which are all pairwise adjacent in $G$).
Our result is even constructive, i.e., we present an algorithm for finding $V'$ and a partitioning
that satisfies these properties.

\begin{theorem}
\label{thm:main-struct}
  Let $(G,\pi)$ be an online graph for which \ff uses $x$ colors. Then, there exists a set $V'\subseteq V$ of size $|V'| = x+q$ with $0\le q \le x-2$, such that $V'$ can be partitioned into $q+1$ non-trivial subsets of vertices, each of which is a clique.
\end{theorem}

Our second contribution is to develop a $1$-consistent and smooth algorithm for online graph coloring with predictions, called \ffp in Subsection~\ref{sec:ffp}. Consider the setting where the algorithm, upon the reveal of a vertex $v$ also obtains a prediction on the color that $v$ should be colored with in an optimal coloring. We give an algorithm that employs \ff with a distinct color palette for each subgraph induced by the set of vertices that obtained the same prediction. By carefully utilizing the aforementioned structural result, we are able to associate the number of colors used by the algorithm with the number of wrong predictions obtained. More specifically, we are able to show that the number of colors used by the algorithm differs from that of an optimal coloring by at most the number of wrong predictions (implying that if the predictions are perfect, the algorithm actually recovers an optimal coloring, even though the quality of the predictions is not a priori known to the algorithm).

\begin{theorem}
\label{thm:consistency}
Assume that \ffp uses $x(G)$ colors for some online graph with predictions $(G,\pi,P)$ whose chromatic number is unknown to the algorithm, then
$x(G)\leq{\eta(G)+\chi(G)}$.
\end{theorem}

Our third contribution is a novel framework for combining different online graph coloring algorithms in Subsection~\ref{sec:rffp}. Earlier frameworks developed for other online problems do not seem to carry over to the online graph coloring problem. Our framework allows us to robustify our algorithm by combining it with a classical algorithm that disregards the predictions. We show that the
number of colors used by the combination of the two algorithms is within a factor of $2$ from that of the best performing of the two on this input instance. This directly implies that this combination is a $2$-consistent, smooth and robust algorithm.

Although in this paper we only use the framework to combine two algorithms, we prove the result for combining any number $t$ of online algorithms (at a loss of $t$ in the competitive ratio). Given an online graph $(G,\pi)$ and an online coloring algorithm $A$, let $A(G)$ denote the number of colors $A$ uses for $G$.

\begin{theorem}
\label{thm:combination}
Let $A_1, A_2, \dots , A_t$ be online graph coloring
algorithms that may or may not make use of the predictions. Then, there exists a meta-algorithm $A$ that combines $A_1,A_2, \ldots , A_t$, such that
for any online graph with predictions $(G,\pi, P)$
\begin{align*}
 A(G) \le t\cdot\min_{1\le i \le t} A_i(G).
\end{align*}
\end{theorem}

The generality of our result allows us to obtain learning augmented algorithms for online graph coloring for several different graph classes in Subsection~\ref{sec:graph-classes}.

\section{A Structural Theorem about \textsc{\textmd{FirstFit}}}
\label{sec:structural}
This section is devoted to proving Theorem~\ref{thm:main-struct}, which establishes a
relationship between the number of colors used by \ff for an online graph $(G,\pi)$ and a partition of
a subset of $V$ into cliques. As mentioned, our proof is constructive and implies an efficient algorithm for finding such a partition.
In the proof we assume that \ff uses $x\ge 2$ colors which are ordered as $c_{0}<c_{1}<\ldots<c_{x-1}$.
We use $N(u)$ to denote the neighbors of a vertex $u$, i.e., the set of vertices that are adjacent to $u$.

\begin{proof}[Proof of Theorem~\ref{thm:main-struct}]
  Let $t_{x-1}\in V$ be a vertex for which \ff uses color $c_{x-1}$. By the definition of \ff vertex $t_{x-1}$ must be adjacent to vertices $t_0,t_1,\ldots, t_{x-2}$, such that $t_i$ is colored with color $c(t_i) = c_i$ for all $i=0,\ldots, x-2$. Let $S = \bigcup_{i=0}^{x-1}\{t_i\}$ and let $N^-(u) = \{w\in N(u)\cap S\ |\ c(w) < c(u)\}$ be its \emph{neighborhood of smaller-colored vertices within $S$}, $\forall u\in S$.

  Note that the set of vertices $V'_0 = \{t_i\in S \ |\  N^-(t_i)= \{t_0,t_1,\dots, t_{i-1}\}\}$ is a clique. Also note that $|V'_0|\ge 2$, since $\{t_{0},t_{x-1}\}\subseteq V'_0$.
  If $V_0'=S$, then the theorem directly follows for $q=0$. So, assume for the remainder of this proof that $V_0' \neq S$. Let $S' =S\setminus V_0'= \{ u_1,u_2,\dots, u_\ell\}$, in which the vertices of $S'$ are ordered by increasing color. We describe an algorithm for partitioning $S'$ into $q$ subsets, with the property that each of these
  subsets of $S'$ together with one distinct vertex from $V\setminus S$ is a clique of size at least two. Note that this implies the theorem since $1 \le |S'| \le x-2$ and thus $1\le q \le x-2$.

  At a high level, the algorithm iterates over the vertices of $S'$ in order of increasing color. And for each such vertex $u_h$, it either identifies a vertex $\beta(u_h)\in V\setminus S$ with $c(\beta(u_h)) < c(u_h)$ such that $(u_h,\beta(h_h))\in E$, and creates a new vertex set $V_h':=\{u_h,\beta(u_h)\}$ which is a clique of size 2, or it adds $u_h$ to a previously created clique $V'_j$ for some $1\le j < h$ to form a larger clique.

  More specifically, for every $h \in
  \{1,...,\ell\}$ let $\alpha(u_h)\in S'$ be a vertex of maximal color with $c(\alpha(u_h))<c(u_h)$ such that $\alpha(u_h)$ is not adjacent to $u_h$, thus $\alpha(u_h)\not\in N^-(u_h)$. And let $\beta(u_h)\in V\setminus S$ be a vertex with $c(\beta(u_h))=c(\alpha(u_h))$ that is adjacent to $u_h$. Note that such vertices must exist: if $\alpha(u_h)$ did not exist, then $u_h$ would be contained in $V_0'$; and if $\beta(u_h)$ did not
  exist, \ff would have assigned a smaller color, namely $c(\alpha(u_h))$, to $u_h$.

  The algorithm proceeds iteratively over the vertices of $S'$ in order of increasing color. For each vertex $u_h\in S'$, if $\beta(u_h)\not\in V_j'$ for all $j$ with $1\le j <h$, then a new clique $V_h'=\{u_h,\beta(u_h)\}$ is created.
  Else, there exists a $j$ with $1\le j < h$ such that $\beta(u_h)\in V_j'$. In this case, set $V_j' := V_j'\cup \{u_h\}$, in other words, add $u_h$ to $V_j'$. We will show that this creates a larger clique. But first note that by the definition of the algorithm each different $\beta(u_h)$
  is added to exactly one such set $V'_j$,  and each such set $V_j'$ contains exactly one specific $\beta(u_h)$. Thus, the output is indeed a partition of $S'\bigcup \cup_{h} \{\beta(u_h) \}$.

  It remains to argue that upon termination, each set $V'_j$ is a clique. We will prove the
  stronger statement that throughout the execution of the algorithm each set $V_j'$ is a clique, and
  consists of $\beta(u_j)$ and a subset of vertices of $S'$ with a color strictly larger than $c(\beta(u_j)) = c(\alpha(u_j))$. This invariant clearly holds upon creation of such a set $V_j'$, since it is created as a clique $\{u_j, \beta(u_j)\}$ and $c(u_j)>c(\beta(u_j))$. Assume that the invariant holds up to some iteration $h-1$. Now consider iteration $h$ during which $u_h$ gets added to $V_j'$. By the definition of the algorithm $\beta(u_h)\in V_j'$, and thus $\beta(u_h)= \beta(u_j)$. And by the definition of $\beta(u_h)$, it is adjacent to $u_h$. It remains to show that $u_h$ is adjacent to all the vertices in $V_j'\setminus\{\beta(u_j)\}$, and that $c(u_h)>c(\beta(u_j))$. The latter directly follows from our ordering and the fact that $c(u_h)>c(u_j)>c(\beta(u_j))$. For the former, recall that $\alpha(u_h)$ is defined as the vertex of $S'$ of maximal color that is not adjacent to $u_h$. In other words, $N^-(u_h)$ contains a vertex of each color strictly between $c(\alpha(u_h))$ and $c(u_h)$, and therefore each vertex of $S'$ with such a color is adjacent to $u_h$. By the induction hypothesis $V_j'$ only contains such vertices (except for the vertex $\beta(u_j)=\beta(u_h)$ whose adjacency to $u_h$ has already been argued).
\end{proof}

\section{Algorithmic Results}
\label{sec:algorithmic-results}

In this section, we focus on deriving and analysing learning augmented algorithms for online graph coloring. We introduce a consistent and smooth algorithm in Subsection~\ref{sec:ffp}, and show how it can be robustified in Subsection~\ref{sec:rffp}.
Finally, in Subsection~\ref{sec:graph-classes} we
argue how it can be used to obtain learning augmented algorithms for online coloring of specific graph classes.

\subsection{\textsc{\textmd{FirstFitPredictions}} (\textsc{\textmd{FFP}})}
\label{sec:ffp}
Throughout this section we assume that the predicted colors are chosen from the set $\{c_0,c_1,c_2,\ldots\}$.
Given an online graph with predictions $(G,\pi,P)$, upon revealing a new vertex $v$ with prediction $P(v) = c_i$, the algorithm \ffp (\ffpshort for short) employs \ff with a distinct color palette associated with $c_i$. We use
$C(i) = \{c_i^0=c_i,c_i^1,c_i^2,\ldots \}$ to denote
the \emph{color palette} associated with color $c_i$,
implying a natural ordering according to the superscripts. Keeping the colors of each such palette distinct enables us to associate the total number of colors used by \ffpshort to the total prediction error.

\begin{tcolorbox}
\underline{\textsc{\textmd{FirstFitPredictions:}}}
When a new vertex $v$ is revealed with prediction $P(v) = c_i$, assign to $v$ the smallest-superscript
eligible color $c\in C(i)$.
\end{tcolorbox}

\ffpshort implies a partition of the vertices of $G$
(and the subgraphs of $G$ induced by the vertices that have been revealed so far)
based on their color palettes.

\begin{definition}
We say that a vertex $v$ belongs to color palette $C(i)$, if it was assigned a color $c\in C(i)$ by \ffpshort (or equivalently, it received the prediction $c_i$). We use $G_i=(V_i,E_i)$ to denote the subgraph of $G$ induced by the set of vertices of color palette $C(i)$.
\end{definition}

Note that an alternative, equivalent description of \ffpshort is that it colors each induced subgraph $G_i$ of $G$ using \ff with color palette $C(i)$. Also note that it is without loss of generality to assume that the color palettes are distinct: every time a specific color is predicted for the first time, one can ``rename" it to a new, unused color (if required) and define the corresponding color palette  accordingly. Finally, note that \ffpshort
does not require any information on the chromatic number $\chi(G)$ of the graph $G$.

We can now relate the number of colors used by \ffpshort in each color palette to the number of prediction errors within that color palette.

\begin{lemma}
\label{lem:consistency}
    Fix an optimal coloring $O\in\mathcal{O}(G)$, let $\eta_i(G)$
    be the number of vertices $v$ of
    color palette $C(i)$ for which $O(v)\neq P(v)$, and let $x_i(G)$ be the number of distinct colors used by \ffpshort for vertices of $C(i)$. Then
    \begin{align*}
        x_i(G) \le \eta_i(G) + 1.
    \end{align*}
\end{lemma}
\begin{proof}
    By the definition of \ffpshort, the set of vertices of color palette $C(i)$ is exactly the set of vertices that received the prediction $c_i$, and therefore exactly the set of vertices of $G_i$. Furthermore, the color assigned to each vertex of $G_i$ is the same as it would obtain, if $G_i$ was given as input to \ff. Therefore, by Theorem~\ref{thm:main-struct}
    a subset of the vertices of $G_i$ can be partitioned into $q+1$ cliques of size at least two, for some integer $q$ with $0 \le q\le x_i(G)-2$. Since all vertices of each such clique received the same prediction, at most one of them could have obtained a correct prediction.

    Let $V_0',V_1',\dots, V_q'$ be the partition of $V'\subseteq V_i$ implied by applying
    Theorem~\ref{thm:main-struct} on $G_i$. Recall that $|V'|=x_i(G)+q$. Since each such $V'_j$ is a clique in $G_i$ and contains vertices that have all obtained the same prediction, at least $|V_j'|-1$ vertices of $V_j'$ must have obtained a wrong prediction, for all $j\in\{0,1,\ldots,q\}$. Summing up over all such disjoint sets we obtain:
    \begin{align*}
        \eta_i(G) \ge \sum_{j=0}^{q} \left( |V_j'|-1 \right) = \sum_{j=0}^{q}|V_j'| - (q+1)  = |V'| - q - 1 = x_i(G) -1,
    \end{align*}
   confirming the statement of the lemma.
\end{proof}

Given Lemma~\ref{lem:consistency}, we are now ready to prove Theorem~\ref{thm:consistency}.

\begin{proof}[Proof of Theorem~\ref{thm:consistency}]
 Let $\chi(G) = k$ and let $c_0,c_1,\dots, c_\ell$ be the distinct predictions assigned by $P$. Note that neither $k$ nor $\ell$ are a priori known to the algorithm (they are only used for the sake of analysis).

 Obviously, $x(G)=\sum_{i=0}^\ell x_i(G)$, where $x_i(G)$ is defined as in Lemma~\ref{lem:consistency}. Note that if $\ell+1 > k$, then there exist at least $\ell+1 -k$ color palettes in which no vertex received a correct prediction; in other words, color palettes for which $ x_i(G)=\eta_i(G)$. Furthermore, by Lemma~\ref{lem:consistency}, for each color palette we have $x_i(G) \leq{\eta_i(G)+1}$, independently of the relationship between $\ell$ and $k$. So, overall at most $k$ color palettes will contribute the additive term ``$+1$" to $x(G)$, and thus we have that $x(G) \le \eta(G) + k = \eta(G) + \chi(G)$.
\end{proof}

Theorem~\ref{thm:consistency} shows that \ffpshort never uses more than $\eta(G)+\chi(G)$ colors for an online graph $G$ with predictions. The next result shows that there exist graphs for
which this amount of colors may indeed be used.

\begin{lemma}
\label{lem:ffp-lb}
    For every integer $k\ge 2$, there exists an online graph with predictions $(G,\pi,P)$ and $\chi(G)=k$ for which \ffpshort uses $x(G) = \eta(G)+k$ colors.
\end{lemma}

\begin{proof}
Considering a fixed integer $k\ge 2$, we will construct a graph $G$ by appropriately connecting $k$ copies of a complete graph on $k$ vertices. Let $G_1,G_2,\ldots, G_k$ be $k$ disjoint copies of a $K_k$, and let $v_i$ be an arbitrary vertex of $G_i$ for $i=1,\dots, k$. Graph $G$ is obtained by adding an edge from $v_1$ to $v_i$ for $i=2,3,\dots, k$. It is easy to check that
$\chi(G)=k$ by using the fact that $\chi(G_i)=k$, and observing that a proper $k$-coloring of $G$ can be obtained by applying suitable permutations of the colors $1,2,\dots, k$ to the copies of the $K_k$.

Fix an arbitrary arrival permutation $\pi$ of the vertices of $G$. Furthermore, assume that every vertex of $G_i$ comes with a prediction of color $c_i$. Since each $G_i$ is a complete
graph on $k$ vertices, exactly one prediction in each such subgraph is correct. So, we have
\begin{align*}
    \eta(G) = k\cdot(k-1).
\end{align*}

On the other hand, \ffpshort will use a different color palette for each $G_i$. Since
each vertex set $V(G_i)$ is a clique of size $k$, \ffpshort will use exactly $k$ colors in each color palette. So, overall
\begin{align*}
 x(G)=k^2=\eta(G)+k.
\end{align*}
\end{proof}

Note that the above construction can easily be extended to having multiple copies of each graph $G_i$, so that $G$ becomes arbitrarily large, while still yielding the same result.

Theorem~\ref{thm:consistency} and Lemma~\ref{lem:ffp-lb} directly imply the following result.

\begin{theorem}
\label{thm:ffp-exact-analysis}
The competitive ratio of \ffpshort is $1 + \frac{\eta(G)}{\chi(G)}$.
\end{theorem}

Theorem~\ref{thm:ffp-exact-analysis} directly implies the $1$-consistency and smoothness of algorithm \ffpshort: if $\eta(G)=0$, then \ffpshort produces a $\chi(G)$-coloring and is therefore optimal; furthermore the number of assigned colors by \ffpshort grows linearly with the prediction error. Nevertheless, \ffpshort is not a robust algorithm. Indeed, for example, suppose we are given a bipartite online graph of order $n$ with predictions $(G,\pi,P)$ such that $\eta(G)=\Theta(n)$. Then \ffpshort would use $\Theta(n)$ colors. But there exist classical online algorithms (without predictions) \cite{[2],[6],Kierstead96} that can color any bipartite graph with $O(\log n)$ colors. In the next section, we present how \ffpshort can be robustified by elegantly combining it with a classical algorithm.

\subsection{\textsc{\textmd{RobustFirstFitPredictions }}(\textsc{\textmd{RFFP}})}
\label{sec:rffp}
A common approach for robustifying a consistent algorithm is to appropriately combine it with a classical algorithm that disregards the predictions. A first such attempt for robustifying \ffpshort would be to switch to some classical online coloring algorithm $A$, once the number of colors used by \ffpshort becomes larger than some predetermined threshold $T$. Recall that $\chi_A(G)$ denotes the number of colors that algorithm $A$ uses for an online graph with predictions $(G,\pi, P)$, where $\pi=v_1,v_2,v_3,\ldots, v_n$. Furthermore, assume that \ffpshort would for the first time use  $T+1$ colors upon arrival of some vertex $v_i$. Then, by switching to a classical online coloring algorithm $A$ (using the same set of colors) for the restriction of $\pi$ to the \emph{suffix-subgraph} $G'$ induced by $\{v_i,v_{i+1},\ldots, v_n\}$, the total number of colors used by the combined algorithm would be at most $T+\chi_A(G')$. Similarly to the deterministic combination result for problems with an underlying
metric~\cite{AntoniadisCE0S20}, this would already give a robust algorithm, if we can assume
that $A$ is weakly monotone in the following sense.

\begin{definition}
Let $A(G,\pi)$ be the number of colors $A$ uses for $(G,\pi)$, where $\pi = v_1, v_2,\dots, v_n$. Let $\pi(i)$ be the suffix $v_i,v_{i+1},\dots, v_n$ of $\pi$, and let $(G(i),\pi(i))$ be the online subgraph of $(G,\pi)$ induced by the vertices in $\pi(i)$. We say that $A$ is \emph{weakly monotone} (resp. \emph{monotone}) if for any $i$, $A(G(i),\pi(i))\leq c\cdot {A(G,\pi)}$ for some constant $c$ (resp. for $c=1$).
\end{definition}

Unfortunately, and perhaps somewhat surprisingly, we are not aware of any weakly monotone
online graph coloring algorithm with a non-trivial guarantee on the number of used colors.
To give some intuition, in Appendix~\ref{app:weak-monotonicity} we present instances showing that both \ff and \bcm~(See \cite{[7]} and~\cite{[2]}) are not weakly monotone, even on specific classes of bipartite graphs  which are known to admit online competitive coloring algorithms.

We are able to circumvent this issue related to the non-monotonicity by reserving a distinct color palette for each algorithm during the combination. This, however, has the side-effect that after switching to an algorithm $A$ in some round $r$, it is possible that upon arrival of a vertex $v$ the algorithm $A$ itself does not increase its number of used colors (by using a color that was already employed before round $r$), but the combined algorithm does. This rules out an expert-setting approach for combining the algorithms (See~\cite{AntoniadisCE0S20,BlumB00} for more information), but we are still able to bound the total number of colors used by the combined algorithm. More specifically, we are able to combine \ffpshort with $t-1$ classical algorithms and obtain an algorithm \emph{\rffp (\rffpshort for short)} which uses a number of colors bounded from above by the expression in Theorem~\ref{thm:combination}.

\begin{proof}[Proof of Theorem~\ref{thm:combination}]
For $1\le i\le n$, let $G(i)$ denote the (online) graph induced by $\{v_1,v_2,\dots, v_i\}$. For any algorithm $B$, let $c(B,i)$ be the color that $B$ assigns to vertex $v_i$.

In the following, we restrict each of the algorithms $A_1,A_2,\ldots, A_t$ to use its own distinct color palette, where we assume a total ordering of the colors within each palette.
Meta-algorithm $A$ is defined as the algorithm that upon arrival of vertex $v_i$ (and the accompanying prediction $P(v_i)$) colors it with color $c(\text{ALG}_i,i)$, where $\text{ALG}_i\in\{A_1,A_2,\ldots, A_t\}$ is an algorithm realizing $\min_{1\le j\le t} A_j(G(i))$.

Note that since each $A_i$ produces a proper coloring and uses a distinct color palette, the resulting coloring after applying $A$ is proper as well. It remains to argue about the number of colors it would use for $G$.

Let \text{ALG} be $\text{ALG}_n$, and for any algorithm $A_i$, let $d(i)$ be the maximal index such that  $A_i(G(d(i))) \le \text{ALG}(G(d(i)))$. Note that by the definition of $A$, it will not use any color from $A_i$'s color palette on vertices $v_{i+1},v_{i+2},\dots, v_n$. Thus, $A$ uses at most $A_i(G(d(i))$ colors from $A_i$'s color palette. Overall, this gives

\begin{align*}
A(G)\leq \sum_{i=1}^t A_i(G(d(i))).
\end{align*}

By the definition of $d(i)$, the above is at most
\begin{align*}
    \sum_{i=1}^t \text{ALG}(G(d(i))).
\end{align*}

Since an online algorithm cannot alter any assigned color, $\text{ALG}(G(j)) \le \text{ALG}(G(j+1))$ for all $j\in\{1,2,\ldots,n-1\}$. This implies $\text{ALG}(G(d(i))) \le \text{ALG}(G)$, which concludes the proof.
\end{proof}

Lemma~\ref{lem:lower bound} in Appendix~\ref{AppB} shows that the result is tight for this meta-algorithm $A$.

In the previous result, we have been combining deterministic algorithms. However,
Theorem~\ref{thm:combination} easily extends to randomized algorithms as well, assuming that one can simulate the execution of all algorithms simultaneously. The following directly follows from Theorem~\ref{thm:combination}, Jensen's inequality and the concavity of the minimum function.

\begin{corollary}
\label{cor:rand-combination}
Let $A_1, A_2, \ldots, A_t$ be randomized algorithms for online graph coloring that may or may not make use of the predictions, and assume that one can simulate the execution of these algorithms simultaneously. Then, there exists a (randomized) meta-algorithm $A$ that combines $A_1,A_2,\ldots, A_t$ and for any online graph $(G,\pi)$
\begin{align*}
 \mathbb{E}(A(G))\le t\cdot \min_{1 \le i \le t}\mathbb{E}(A_i(G)).
\end{align*}
\end{corollary}

Theorem~\ref{thm:combination} implies that we can combine \ffpshort with a $c$-competitive classical algorithm (perhaps on a specific class of input graphs) and obtain a $2\min\{1+\frac{\eta(G)}{\chi(G)},c\}$-competitive algorithm. Such an algorithm attains a consistency of $2$, and is at the same time smooth (the number of used colors grows linearly with the prediction error) and robust (it is $2c$-competitive, independently of the prediction quality).

Assume that the (learning augmented) algorithm has knowledge that the input graph belongs to a specific class of graphs such that (i) all graphs of this class have chromatic number $k$ and (ii) there exists a classical  online algorithm that is online competitive on this class, with function $f(\cdot)$. In that specific case, a slight adaptation in the proof of Theorem~\ref{thm:combination} even gives us a $1$-consistent algorithm that is at the same time robust.

\begin{corollary}
    \label{cor:2combination}
    For some fixed $k$, assume that the algorithm \ffpshort is aware that the input graph belongs to a class $\mathcal{C}$ of graphs such that $\chi(G)=k$
    for all
    $G\in\mathcal{C}$.
    Moreover, assume that
    a classical online algorithm $A_1$ is
    known
    for all graphs of class $\mathcal{C}$.
    Then, there exists a meta-algorithm $A'$ that combines \ffpshort with $A_1$, and for any online graph $(G,\pi,P)\in\mathcal{C}$,
    \begin{align*}
        {A'}(G) &\le 3\min\{\ffpshort(G), A_1(G)\} \text{ if $\eta(G)>0$}, \text{ and}\\
        {A'}(G) &= k \text{ otherwise.}
    \end{align*}
\end{corollary}

\begin{proof}
    Let $A'$ be the meta-algorithm that colors each vertex $v$ with $P(v)$, until upon the reveal of some $v_i$, either $k+1$ distinct colors are predicted so far, or coloring $v_i$ with $P(v_i)$ would result in an improper coloring. For the remaining sequence $\pi(i)$, $A'$ is defined precisely as $A$ in the proof of Theorem~\ref{thm:combination}. If $\eta(G)=0$, then $A'$ will follow the predictions on the whole input graph, and thus produce a proper $k$-coloring. Otherwise, by applying Theorem~\ref{thm:combination} on \ffpshort and $A_1$, we obtain that $A'$ uses at most $2\min\{\ffpshort(G), A_1(G)\}$ colors on the sequence $\pi(i)$. Combining this with the facts that $A'$ uses at most $k$ distinct colors on $v_1,v_2,\dots, v_{i-1}$ and that, since \ffpshort and $A_1$ produce proper colorings, both $\ffpshort(G) \ge k$ and $ A_1(G) \ge k$, we obtain the required conclusion.
\end{proof}

\subsection{Results for Specific Classes of Graphs}
\label{sec:graph-classes}

Given that the input graph belongs to a specific graph class (and this is known to the algorithm
a priori), we can obtain more refined results. In this section, we review some interesting cases for which classical (deterministic) online algorithms are known.

\subsubsection{Bipartite Graphs}
Although the chromatic number of bipartite graphs is $2$, for any deterministic online coloring algorithm $A$, there exists an online bipartite graph on $n$ vertices that forces $A$  to use at least $2\log n -10$ colors~\cite{[5]}. This result is essentially tight, given that there exists a simple online coloring algorithm guaranteeing a coloring with at most $2\log n +1$ colors on any bipartite graph on $n$ vertices~\cite{[6]}. This means that there is no competitive algorithm for bipartite graphs. However, there are (online) competitive algorithms for specific subclasses of bipartite graphs. Namely, for $P_4$-free bipartite graphs \ff is known to be optimal~\cite{[1]} whereas for $P_5$-free bipartite graphs it uses at most three colors~\cite{[11]}. Online algorithm \bcm guarantees a coloring with four colors on any  $P_5$-free bipartite graph~\cite{[2]}. The problem becomes significantly more difficult on $P_6$-free bipartite graphs where, as mentioned, no online algorithm can guarantee a coloring with constantly many colors. However, \bcm is online competitive for such graphs, and  $\chi_{\bcm}(G)\le 2\chi_{OL}(G)-1$ for any  $P_6$-free bipartite graph~\cite{[2],[10],[3]}. For $P_7$-free, $P_8$-free and $P_9$-free bipartite graphs, an algorithm that builds upon \bcm is known to be online competitive and uses at most $4\chi_{OL}(G)-1$, $3(\chi_{OL}+1)^2$ and $6(\chi_{OL}(G)+1)^2$ colors, respectively~\cite{[10],[3]}. Applying Corollary~\ref{cor:2combination} to \ffpshort and an algorithm of the respective graph class gives the following theorem.

\begin{theorem}
    There exist (different) algorithms for online coloring bipartite graphs with predictions obtaining a competitive ratio of $1$ if $\eta(G) = 0$, and $3\min\{\frac{\eta(G)}{2}+1,X\}$ otherwise,  where $X$ is
    \begin{itemize}[topsep=0pt,parsep=0pt,partopsep=0pt]
        \item $\Theta(\log n)$ for general bipartite graphs,
        \item  $\chi_{OL}(G)-\frac{1}{2}$ for bipartite $P_6$-free graphs,
        \item  $2\chi_{OL}(G)-\frac{1}{2}$ for bipartite $P_7$-free graphs,
        \item $\frac{3(\chi_{OL}(G)+1)^2}{2}$ for bipartite $P_8$-free graphs, and
        \item $3(\chi_{OL}(G)+1)^2$ for bipartite $P_9$-free graphs.
    \end{itemize}
\end{theorem}

\subsubsection{Other graph classes}
Besides bipartite graphs, the (online) graph coloring problem has been extensively studied for several other graph classes. Among them, for instance, are chordal graphs, intersection graphs, $d$-inductive graphs (also known as $d$-degenerate graphs), graphs with bounded treewith and graphs with forbidden induced subgraphs.
A \emph{chordal graph} is a simple graph, in which every cycle of length greater than three has a chord. It is known that \ff colors every chordal graph $G$ with $\chi(G)=d$ using $O(d \cdot \log n)$ colors~\cite{[16]}, which is best possible for any deterministic online algorithm~\cite{[12]}.
The intersection graph of a set of disks in the Euclidean plane is the graph having a vertex for each disk and an edge between two vertices if and only if the corresponding disks overlap. A graph $G$ is called a \emph{disk graph}, if there exists a set of disks in the Euclidean plane whose intersection graph is $G$. \ff is also $\Theta(\log n)$-competitive on disk graphs, and it is best possible for any deterministic online algorithm~\cite{[13]}.

A graph is \emph{$d$-inductive} (or $d$-degenerate), if it can be reduced to $K_1$ by repeatedly deleting vertices of degree at most $d$. Intuitively, the \emph{treewidth} of a graph $G$ is an integer denoting how far $G$ is from being a tree. More formally, a \emph{tree decomposition} of $G=(V,E)$ is a tree $T$ with vertices $Y_1,Y_2,\dots, Y_n$ where $Y_i\subseteq V$ for all $i$ and $\bigcup_i V_i = V$, such that all $Y_i$'s that contain a vertex $v\in V$ form a connected subtree in $T$, and furthermore for all $e=(v,w)\in E$ there exists a $Y_i$ such that $v\in Y_i$ and $w\in Y_i$. The \emph{width} of a tree decomposition is the size of its largest set $Y_i$ minus one, and the \emph{treewidth} of graph $G$ is defined as the minimum width among all possible tree decompositions of $G$. \ff colors any $d$-inductive graph and any graph of treewidth $d$ using $O(d \cdot \log n)$ colors~\cite{[12]}. This is best possible for both classes, by the aforementioned lower bound on chordal graphs and the fact that every chordal graph $G$ is $(\chi(G)-1)$-inductive and has treewidth $\chi(G)-1$~\cite{Bodlaender93}.

So altogether applying Theorem~\ref{thm:combination} on \ffpshort and \ff, we get the following result for chordal graphs, disk graphs and $d$-inductive graphs, as well as for graphs of treewidth $d$.

\begin{theorem}
    There exist (different) algorithms for online coloring chordal graphs, $d$-inductive graphs, graphs of treewidth $d$ and disk graphs with predictions obtaining a competitive ratio of
    $1$ if $\eta(G) = 0$, and $2\min\{\frac{\eta(G)}{\chi(G)}+1,X\}$ otherwise, where $X=O(\log n)$ is the competitive ratio of the respective classical online algorithm.
\end{theorem}

The complementary notion of a clique is an \emph{independent set}, that is a set $S\subseteq V$ such that no two vertices of $S$ are adjacent. An independent set of size $s$ is denoted by $I_s$. \ff is known to achieve a competitive ratio of $t-1$ on $K_{1,t}$-free graphs where $t\geq3$~\cite{[14]}, and there exist classical online algorithms which are $\frac{s}{2}$-competitive on $I_s$-free graphs~\cite{[15]}. Therefore, applying Theorem~\ref{thm:combination} to \ffpshort and the respective algorithm for each of these two classes of graphs gives the following theorem.

\begin{theorem}
There exist (different) algorithms for online coloring $I_s$-free graphs and $K_{1,t}$-free graphs for $t\geq3$ with predictions obtaining a competitive ratio of $1$ if $\eta(G) = 0$, and
$2\min\{\frac{\eta(G)}{\chi(G)}+1,X\}$ otherwise, where $X$ is
    \begin{itemize}[topsep=0pt,parsep=0pt,partopsep=0pt]
        \item  $t-1$ for $K_{1,t}$-free graphs,
        \item $\frac{s}{2}$ for $I_s$-free graphs.
    \end{itemize}
\end{theorem}

\section{Discussion}
In this paper, we presented a simple learning augmented algorithm for graph coloring that is $2$-consistent, smooth and robust. When the chromatic number of the input graph is known, a $1$-consistent, smooth and robust algorithm is obtained. It would be interesting to investigate whether a learning augmented algorithm with a consistency better than $2$ is possible, when the chromatic number of the input graph is not known to the algorithm.

Furthermore, all presented algorithms are smooth and the number of used colors grows linearly with the prediction error. It is an open question whether there exist any learning augmented algorithms which achieves the same consistency, but with a better dependence on the prediction error.

\bibliographystyle{plain}
\bibliography{references}

\pagebreak
\appendix
\section{Monotonicity Counterexamples}
\label{app:weak-monotonicity}

\begin{lemma}
\ff is not weakly monotone.
\end{lemma}

\begin{proof}
 Consider a bipartite graph $G$ (see Figure $1$) with bipartition $X=\{v_{1},\ldots,v_{n}\}$ and $Y=\{u_{1},\ldots,u_{n}\}$ and $E(G)=\{v_{1}u_{1}\}\cup\{v_{i}u_{j}|i\neq{j}, i,j= 1,2,\ldots,n\}$ be the edge set of $G$. The vertices arrive one by one in the following order:\\
 $$v_{1},u_{1},v_{2},u_{2},\ldots,v_{n},u_{n}.$$
 Now we take $v_{1}$ and $u_{1}$ away from $G$ and let $G'$ denote the remaining subgraph. The order of arrival of the vertices in $G'$ stays the same as that in $G$. It is easy to verify that \ff would only use two colors on $G$ whereas $n-1$ colors on  $G'$.
\end{proof}

\begin{figure}
\begin{center}
   \begin{tikzpicture}[scale=0.6,auto,swap]
   \tikzstyle{blackvertex}=[circle,draw=black]
   \node [label=above:$v_{1}$ ,blackvertex,fill=black,scale=0.3] (a1) at (-13,2) {};
   \node [label=above:$v_{2}$ ,blackvertex,fill=black,scale=0.3] (a2) at (-11,2) {};
   \node [label=above:$v_{3}$ ,blackvertex,fill=black,scale=0.3] (a3) at (-9,2) {};
   \node [label=below: $\cdots$] (a4) at (-7.5,2.5) {};
   \node [label=above:$v_{n}$ ,blackvertex,fill=black,scale=0.3] (a5) at (-6,2) {};

   \node [label=below:$u_{1}$ ,blackvertex,fill=black,scale=0.3] (a6) at (-13,-1) {};
   \node [label=below:$u_{2}$ ,blackvertex,fill=black,scale=0.3] (a7) at (-11,-1) {};
   \node [label=below:$u_{3}$ ,blackvertex,fill=black,scale=0.3] (a8) at (-9,-1) {};
   \node [label=below: $\cdots$] (a9) at (-7.5,-0.5) {};
   \node [label=below:$u_{n}$ ,blackvertex,fill=black,scale=0.3] (a10) at (-6,-1) {};
   \node [label=right: $G$] (a11) at (-10,-2) {};
   \draw [black,thick] (a1)--(a6);
   \draw [black,thick] (a1)--(a7);
   \draw [black,thick] (a1)--(a8);
   \draw [black,thick] (a1)--(a10);
   \draw [black,thick] (a2)--(a6);
   \draw [black,thick] (a2)--(a8);
   \draw [black,thick] (a2)--(a10);
   \draw [black,thick] (a3)--(a6);
   \draw [black,thick] (a3)--(a7);
   \draw [black,thick] (a3)--(a10);
   \draw [black,thick] (a5)--(a6);
   \draw [black,thick] (a5)--(a7);
   \draw [black,thick] (a5)--(a8);

   \tikzstyle{blackvertex}=[circle,draw=black]
   \node [label=above:$v_{2}$ ,blackvertex,fill=black,scale=0.3] (a2) at (-0.5,2) {};
   \node [label=above:$v_{3}$ ,blackvertex,fill=black,scale=0.3] (a3) at (1.2,2) {};
   \node [label=below: $\cdots$] (a4) at (2.2,2.5) {};
   \node [label=above:$v_{n}$ ,blackvertex,fill=black,scale=0.3] (a5) at (3.5,2) {};
   \node [label=below:$u_{2}$ ,blackvertex,fill=black,scale=0.3] (a7) at (-0.5,-1) {};
   \node [label=below:$u_{3}$ ,blackvertex,fill=black,scale=0.3] (a8) at (1.2,-1) {};
   \node [label=below: $\cdots$] (a9) at (2.2,-0.5) {};
   \node [label=below:$u_{n}$ ,blackvertex,fill=black,scale=0.3] (a10) at (3.5,-1) {};
   \node [label=right: $G'$] (a11) at (1.2,-2) {};
   \draw [black,thick] (a2)--(a8);
   \draw [black,thick] (a2)--(a10);
   \draw [black,thick] (a3)--(a7);
   \draw [black,thick] (a3)--(a10);
   \draw [black,thick] (a5)--(a7);
   \draw [black,thick] (a5)--(a8);

\end{tikzpicture}
\caption{\ff is not weakly monotone.}
\label{fig:ff}
\end{center}
\end{figure}
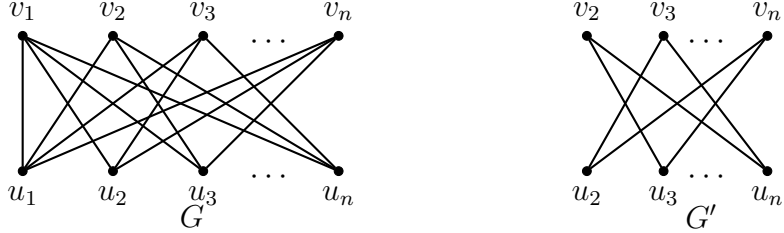

\begin{lemma}
\label{lem:bcm-not-monotone}
 \bcm is not weakly monotone.
\end{lemma}

\begin{proof}
In~\cite{[3],[10]} the authors present an online bipartite $P_6$-free graph $(G=(V_1\cup V_2,E),\pi)$ for which \bcm requires $O(\log n)$ colors. We construct an online $P_6$-free bipartite graph $(G',\pi')$ as follows. Start with $G$ and add two vertices $v$ and $v'$. Add edges connecting $v$ to $v'$ and each of the vertices of $V_1$, and edges connecting $v'$ to each of the vertices in $V_2$. The sequence of vertex arrivals is $\pi' = v,v', \pi$. By construction $G'$ is bipartite. It is also $P_6$-free, suppose on the contrary there is an induced $P_6$ in $G'$. Since $G$ is $P_6$ free, such an induced $P_6$ would have to contain $v$ or $v'$ and at least two vertices from $V_1$ as well as at least two vertices of $V_2$. But $v$ is adjacent to all the vertices in $V_1$ and $v'$ is adjacent to all the vertices in $V_2$, so the induced subgraph cannot be a path, contradiction.

Note that by the definition of $(G',\pi)$ in each step of the (any) algorithm the currently revealed subgraph will always be connected. By the definition of \bcm (see~\cite{[2]}), it only uses $2$ colors on any bipartite $P_6$-free graph that remains connected throughout the execution of the algorithm. In other words, \bcm uses $2$ colors on $(G',\pi')$ but requires $O(\log n)$ colors on $(G,\pi)$ (where $\pi$ is a suffix of $\pi'$ and $G$ is the subgraph of $G'$ induced by the vertices in $\pi$).
\end{proof}
In~\cite{[3],[10]} extensions of \bcm were presented for $P_8$-free and $P_9$-free bipartite graphs. The construction in the proof of Lemma~\ref{lem:bcm-not-monotone} can be used to show that these algorithms are also not weakly monotone. We defer to the full version for a more extensive discussion.

\section{Lower Bound for the Meta-Algorithm}\label{AppB}
\begin{lemma}~\label{lem:lower bound}
    There exist deterministic algorithms $A_1, A_2, \dots, A_t$ and an online graph with predictions $(G,\pi, P)$ such that the meta-algorithm $A$ that combines $A_1, A_2 \dots, A_t$ in the way described in the proof of Theorem~\ref{thm:combination} uses exactly $t\cdot\min_{1\le i \le t} A_i(G)$ colors for $(G,\pi, P)$.
\end{lemma}

\begin{proof}
    Let the color palette associated with $A_i$ be $C(i)=\{c_i^{0}, c_i^{1}, \dots, c_i^{t-1}\}$ for $1 \le i\le t$. Let $G=tK_1$ be given to the algorithm $A$ in an arbitrary order.
    Let $A_i$, for each $1\le i \le t$ be such that it uses $c_i^0$ to color the first $i$ vertices of $G$ and assigns a different color to each of the remaining vertices of $G$. Without loss of generality, we can assume that $A$ follows $A_1$ in the first iteration and switches from $A_i$ to $A_{i+1}$ after the $i$th iteration (the algorithm $A$ satisfies the description from the proof of Theorem~\ref{thm:combination}). The resulting coloring is trivially a proper coloring, since $G$ contains no edges. The lemma follows from the fact that $A$ switches to using a different algorithm after each round, and ends up using $t$ colors for $(G,\pi, P)$ while algorithm $A_t$ only uses $1$ color for $(G,\pi,P)$.
\end{proof}

\end{document}